\documentclass[a4paper,12pt,reqno,superscriptaddress,showkeys,nofootinbib]{revtex4}
\usepackage[centertags]{amsmath}
\usepackage{amsfonts}
\usepackage{amssymb}
\usepackage{amsthm}
\usepackage{newlfont}
\usepackage{stmaryrd}
\usepackage{mathrsfs}
\usepackage{euscript}
\usepackage{siunitx}
\usepackage{algorithm2e}
\usepackage{graphicx,subcaption}
\usepackage{enumitem}
\usepackage{natbib}
\usepackage{graphicx}
\usepackage{color}
\usepackage{floatrow}
\usepackage{caption}
\usepackage{hyperref}
\usepackage{hhline}
\usepackage{hyperref}
\usepackage{cleveref}
\usepackage{import}

\theoremstyle{plain}
  \newtheorem{theorem}{Theorem}[section]
  
  \newtheorem{proposition}[theorem]{Proposition}
  \newtheorem{lemma}[theorem]{Lemma}
\theoremstyle{definition}

\theoremstyle{remark}

  \newtheorem{example}[theorem]{Example}

\numberwithin{equation}{section}

\let\ve=\varepsilon

\newcommand{\opunit}{\text{1}\kern-0.22em\text{l}}

\DeclareMathAlphabet{\mathpzc}{OT1}{pzc}{m}{it}

\newcommand{\fig}{Fig.\;}

\newcommand{\id}{\textrm{d}}

\usepackage{floatrow}
\usepackage{caption}

\begin{document}
\title{Close-to-equilibrium heat capacity}
\author{Faezeh Khodabandehlou and Christian Maes\\
{\it Department of Physics and Astronomy, KU Leuven, Belgium}} 
\email{faezeh.khodabandehlou@kuleuven.be}

\begin{abstract}
Close to equilibrium, the excess heat governs the static fluctuations. 
We study the heat capacity in that McLennan regime, {\it i.e.,} in linear order around equilibrium, using an expression in terms of the average energy that extends the equilibrium formula in the canonical ensemble. It is derivable from an entropy and it always vanishes at zero temperature. Any violation of an extended Third Law is therefore a nonlinear effect.
\end{abstract}
\maketitle

\section{Introduction}
Knowing a physical system signifies that we can also correctly predict or at least describe its reaction to an arbitrary stimulus.  Such knowledge is summarized in susceptibilities or response functions, to begin with the linear (or perturbative) regime.  The present paper deals with thermal response in the linear regime around equilibrium.  That response gets quantified by the notion of heat capacity, which for the present purpose is studied in the McLennan ensemble \cite{JAM,McL}.  In that sense, it is a continuation of earlier work where the first expressions of nonequilibrium heat capacity have been given, \cite{epl}.  \\

In recent papers \cite{activePritha,sim,prithaarray,jchemphys,drazin}, building on suggestions from steady-state thermodynamics \cite{oono,kom2,Komatsu_2009}, the notion of heat capacity was generalized to driven and active systems.  There, a quasipotential captures the quasistatic excess heat, \cite{mathnernst,jchemphys,epl,Pe_2012}, which is the analogue of energy changes under detailed balance processes. One thus wonders whether that quasipotential also governs the occupation statistics of energy levels.  The answer is yes in the close-to-equilibrium regime. \\
The analysis continues by discussing the close-to-equilibrium expression for the heat capacity.  As it is a simplification with respect to the situation far from equilibrium, we feel invited to understand its properties in more detail. At this moment, we may even wonder whether that simplified formula for the heat capacity, which is strictly speaking only valid close to equilibrium, gives zero at absolute zero (extended Nernst postulate).  The answer is again yes. Failure of Third Law-like behavior is, therefore, a nonlinear effect.\\ 

 The quasipotential is defined in the next Section \eqref{def}.  In Section \ref{mldis}, we remind the reader of the McLennan ensemble, which gives an extension of the Boltzmann-Gibbs weights to the close-to-equilibrium regime.  In particular, in the McLennan distribution, we see the quasipotential as a replacement for the energy function in equilibrium distributions.\\
In Section \ref{heatcapacity} we give the derivation of a formula for the heat capacity in the McLennan ensemble that has been presented before in \cite{epl}.  That formula is simpler than the general nonequilibrium expression and is worth independent study. \\
Sections \ref{bv}--\ref{example}  investigate the properties of that McLennan heat capacity, in comparison with the (full) heat capacity and as derived from the corresponding quasipotentials.  In particular, we see the range of validity of the simplified McLennan expression of the heat capacity and how it always satisfies an extended Third Law (even though the formula is strictly speaking not valid at very low temperatures).\\
We end by summarizing the findings in the conclusions of Section \ref{con}, and we add a more technical Appendix about the graphical solution of the Poisson equation, \cite{pois}.

\section{Definition of quasipotential}\label{def}
Suppose we have a physical system that can exchange energy with a heat bath at inverse temperature $\beta^{-1}$.  In the usual weak-coupling limit \cite{prigo,naka,zwanzig1}, the description may proceed in terms of a Markov jump process $X_t$ for the system state at time $t$. That is specified from transition rates $k(x,y)$ for the jump between system states $x\rightarrow y$.
Under local detailed balance \cite{ldb}, a further interpretation is possible: the heat $q(x,y)= - q(y,x)$ released to the bath during $x\rightarrow y$ is found from taking the $\log$-ratio between forward and backward transitions,
\begin{equation}\label{ldb}
q(x,y) = \frac 1{\beta} \log\frac{k(x,y)}{k(y,x)}
\end{equation}
In general, we have the First Law
\begin{equation}\label{1st}
q(x,y) = E(x) - E(y) + b(x,y)
\end{equation}
where $b(x,y) = -b(y,x)$ is the work done on the system in the transition $x\rightarrow y$.  While antisymmetric, the $b(x,y)$ need not be the difference of a unique potential $\cal F$.  The nonexistence of a state function $S = \beta(E - \cal F)$ so that for all transitions, $\beta\, q(x,y) = S(x) - S(y)$, signifies the violation of the condition of global detailed balance.\\
The heat flux equals
\[
\dot{q}(x) = 
\sum_y k(x,y) q(x,y)
\]
Assuming a finite number of system states and the irreducibility of the graph of transitions \cite{grimm}, there is a unique stationary distribution $\rho^s>0$, solution of the stationary Master equation $\sum_y[\rho(x) k(x,y) - k(y,x) \rho(y) ]=0, \forall x$.  We denote the (constant) stationary heat flux by
\[
\dot q^s = \sum_x \dot{q}(x)\rho^s(x) =  \lim_{t\uparrow \infty} \big\langle \dot{q}(X_t)\,|\,X_0=x \big\rangle 
\]
Finally, the quasipotential is defined as
\begin{equation}\label{qs}
V(x) = \int_0^\infty \id t \, \big[\langle \dot{q}(X_t)\,|\,X_0=x\rangle -  \dot q^s\big]
\end{equation}
where we time-integrate the difference between the expected heat flux at time $t$ (when starting at time zero from state $x$) and the stationary heat flux. Equivalently, we can rewrite
\[
V = \int_0^\infty \id t \,e^{tL}[\dot{q}-\dot{q}^s]
\]
in terms of the backward generator $L$, for which $Lf(x) =\sum_y k(x,y)[f(y)-f(x)]$ on arbitrary functions $f$.  Hence, alternatively to \eqref{qs},
$V$ is the unique solution to the Poisson equation
\begin{equation}\label{pe}
    LV (x)= \dot{q}^s - \dot{q}(x)
\end{equation}
when imposing that its stationary expectation $\langle V\rangle^s=\sum_xV(x)\rho^s(x)=0$ vanishes.\\

Using \eqref{qs}--\eqref{pe}, we define a probability distribution
\begin{equation}\label{mc}
\rho_V(x) = \frac 1{\cal Z} e^{-\beta V(x)}
\end{equation}
Note that the quasipotential $V$ itself typically depends on the inverse temperature $\beta$, and there is no {\it a priori} reason for it to be local in the sense of a potential. In the case where $b\equiv 0$ in \eqref{1st}, the quasipotential $V(x) = E(x) - \langle E\rangle_\text{eq}$ equals the energy function (up to a constant), and \eqref{mc} becomes the Gibbs probability  $\rho_\text{eq} \sim \exp[-\beta E]$ for calculations in the canonical ensemble.  It turns out that \eqref{mc} remains meaningful in the close-to-equilibrium regime, which is the subject of the following section.  The fact is that \eqref{mc} is the correct stationary distribution to first order around equilibrium; that is the McLennan distribution.

\section{McLennan distribution}\label{mldis}
We start from a reference dynamics that satisfies detailed balance, having transition rates $k_\text{eq}(x,y)$ with
\[
\log\frac{k_\text{eq}(x,y)}{k_\text{eq}(y,x)} =  \beta[E(x) - E(y)],\qquad k_\text{eq}(y,x)\,\rho_\text{eq}(y)=k_\text{eq}(x,y)\,\rho_\text{eq}(x)
\]
for energy function $E$. Its backward generator is $L_\text{eq}$ and the equilibrium distribution is $\rho_\text{eq}$.\\
Suppose now that the Markov process under consideration  is a small and smooth perturbation, in the sense that its transition rates 
\begin{equation}\label{rate}
    k(x,y)=k_\text{eq}(x,y)(1+\ve[\psi_\beta(x,y) + \frac{\beta}{2}w(x,y)]) + O(\ve^2)
\end{equation}
are smooth in $\ve \ll 1$, parametrized by symmetric $\psi_\beta(x,y) = \psi_\beta(y,x)$ and antisymmetric $w(x,y) = -w(y,x)$ not depending on $\beta$.  The corresponding backward generator is $L = L_\text{eq} + 
\ve L_1 + O(\ve^2)$.\\
In that regime, we have for \eqref{1st} that the work $b(x,y) = \ve\, w(x,y) + O(\ve^2)$ has small amplitude $\ve$, and
\begin{equation}\label{wve}
    q(x,y) = E(x) - E(y) + \ve w(x,y) + O(\ve^2)
\end{equation}

Write
\begin{equation}\label{soul}
F(x) = - \ve\,\sum_yk_\text{eq}(x,y)\,w(x,y)
\end{equation}
which is automatically centered, meaning that $\langle F\rangle_\text{eq} =0$ because $w(x,y)=-w(y,x)$ implies that $ \sum_{x,y}k_\text{eq}(x,y)\,w(x,y)\rho_\text{eq}(x)
=0$.\\

We define the McLennan distribution as
\begin{equation}\label{mca}
    \rho_\text{ML}(x) = 
    \rho_\text{eq}(x)\;( 1 - \beta W(x))
\end{equation}
where $W=O(\ve)$ solves the Poisson equation \eqref{pe} in the form
\begin{equation}\label{wp}
L_\text{eq}W=  F,\qquad \langle W\rangle_\text{eq} =0
\end{equation}
The distribution \eqref{mca} (with expectations denoted by $\langle \cdot\rangle_\text{ML}$) is automatically normalized since we have required $\langle W\rangle_\text{eq} =0$ (making $W$ unique).  The positivity $\rho_\text{ML}(x)>0$ follows when the $\ve\,\beta w(x,y)$ are small enough.  Note that $\langle W \rangle_\text{ML} = O(\ve^2)$, while $\langle E \rangle_{ML} - \langle E \rangle_\text{eq} = -\beta \langle EW\rangle_\text{eq}$.\\
  The first and most important property of the McLennan distribution \eqref{mca}, is that it correctly describes the linear response regime around equilibrium, \cite{JAM,McL}.  In other words, the distribution \eqref{mca} is ``correct'' close-to-equilibrium (under the usual assumption of weak coupling, {\it etc}).  In that regime, the static fluctuations are described exactly by the corresponding quasipotential: the relation between \eqref{mca} and \eqref{mc} is that $\rho_\text{ML}=\rho_{V_\text{ML}} + O(\ve^2)$ where $V_\text{ML} =  E  - \langle E\rangle_\text{ML}  + W$ and $W$ solves \eqref{wp}.  In other words, $V = V_\text{ML} + O(\ve^2)$.  Note in particular that the time-symmetric $\psi_\beta(x,y)$ in \eqref{rate} have no contribution to the stationary distribution to linear order in $\ve$. We repeat the formal derivation (which can be skipped at first reading):\\

The Gibbs distribution $\rho_\text{eq}$ satisfies the equilibrium Master equation $L_\text{eq}^\dagger\rho_\text{eq} =0$ for Markov generator $L_\text{eq}^\dagger$.  The close-to-equilibrium process has generator $L^\dagger = L_\text{eq}^\dagger + \ve L_1^\dagger$.  For $\rho_\text{ML}= \rho_\text{eq}( 1 - \beta W)$, we demand that
\[
(L_\text{eq}^\dagger + \ve L_1^\dagger)\rho_\text{ML} = -\beta\,L_\text{eq}^\dagger (\rho_\text{eq} \,W) + \ve L_1^\dagger\rho_\text{eq} + O(\ve^2) = O(\ve^2)
\]
In other words we try to find $W$ in $\rho_\text{ML}= \rho_\text{eq}( 1 - \beta W)$ which is stationary up to linear order in $\ve$.  The previous equality then implies that we need to solve
\begin{equation}\label{cr}
L_\text{eq}^\dagger (\rho_\text{eq} W) =\frac{\ve}{\beta} \,L_1^\dagger \rho_\text{eq} 
\end{equation}
In equilibrium, $L_\text{eq}^\dagger (\rho_\text{eq} W)\,(x) =  \rho_\text{eq}(x)\, L_\text{eq} W (x)$, or, for arbitrary functions $f$,
\[
\sum_x \rho_\text{eq}(x)\,W(x)\;L_\text{eq}\,f (x)\,  =  \sum_x \rho_\text{eq}(x) \,f (x)\; L_\text{eq} W(x) 
\]
Therefore,  $W$ solves
\begin{equation}\label{cre}
L_\text{eq} W = G, \qquad G = \frac{\ve}{\beta\rho_\text{eq}} L_1^\dagger \rho_\text{eq}
\end{equation}
where $G$ satisfies the centrality condition $\beta\,\langle G\rangle_\text{eq}=  \ve \sum_x L_1^+\rho_\text{eq}(x) = 0$.
In other words, $ 1 - \frac{\rho_\text{ML}}{\rho_\text{eq}} = \beta W$ where $W$ solves \eqref{cre}.\\
The point now is that the source $G$ in \eqref{cre} equals the source $F$ in \eqref{soul}, and hence is directly related to the irreversible work.  To see that, we  remember \eqref{rate},
so that
\[
L_1f(x) = \sum_y k_\text{eq}(x,y)[\psi_\beta(x,y) + \frac{\beta}{2}w(x,y)](f(y)-f(x))
\]
and
\begin{align*}
    L_1^+\rho_\text{eq}(x)=&-\sum_y k_\text{eq}(x,y)[\psi_\beta(x,y) + \frac{\beta}{2}w(x,y)]\rho_\text{eq}(x)\\
    &\,\,+\sum_yk_\text{eq}(y,x)[\psi_\beta(y,x) + \frac{\beta}{2}w(y,x)]\rho_\text{eq}(y)\\
    &=-\beta \, \rho_\text{eq}(x)\, \sum_yk_\text{eq}(x,y)\,w(x,y)
\end{align*}
Substituting in \eqref{cre}, we find
\begin{equation}\label{sou}
G(x) =-\ve \sum_yk_\text{eq}(x,y)\,w(x,y) = F(x)
\end{equation}
We conclude that $W$ indeed satisfies the Poisson equation \eqref{wp}, with source \eqref{soul}.\\
Putting $V_\text{ML} = E  - \langle E\rangle_\text{ML}  + W= E  - \langle E\rangle_\text{eq}  + \beta\langle EW\rangle_\text{eq} + W$,   we have
$\langle V_\text{ML}  \rangle_\text{ML} = O(\ve^2)$ and
\begin{eqnarray}\label{cal}
L_\text{eq}[V_\text{ML} -E] &=& -\ve \sum_yk_\text{eq}(x,y)\,w(x,y)\,\implies\\
L_\text{eq}V_\text{ML}  &=&  \sum_yk_\text{eq}(x,y)[E(y) - E(x) -\ve\,w(x,y)]\implies \\
LV_\text{ML}  &=& \sum_yk(x,y)[E(y) - E(x) -\ve\,w(x,y)] + O(\ve^2)
\end{eqnarray}
since  $\ve L_1V_\text{ML}  = \ve L_1 E + O(\ve^2)$ and $LW = L_\text{eq}W + O(\ve^2)$.
That concludes the argument.\\

In Section \ref{example}  we collect two examples for illustration of more technical points. Here we add a more physical example.

\begin{example}[McLennan distribution for heat conduction]
We consider the Kipnis-Marchioro-Presutti (KMP)  model for heat conduction on a lattice interval $\{1,..,N\}$, \cite{kmp,kmp2023}. Each end is imagined connected to a thermal bath but at a different temperature, left and right.\\
The dynamical variables are the energies $x_i$ at site $i$ in configuration $x=(x_1,...,x_N)$. The transitions $x\rightarrow y$ occur between neighboring sites (in the bulk) and at the edges.  More specifically, in the bulk,  for two neighbors $i$ and $i+1$ and with  uniformly chosen $c\in [0,1]$,  we put $y_i= c \,(x_i+x_{i+1})$ and $y_{i+1}= (1-c)\,(x_i+x_{i+1})$, which amounts to a local redistribution of the energies.  At the edges, left and right, a random amount of energy enters or leaves the system, weighted according to an equilibrium distribution with left respectively, right temperature.\\ 

The generator for the  KMP-process can therefore be written as $L=L_r+\sum_{i=1}^{N-1} L_i+L_{\ell}$.
For the bulk,
\[
L_i f\,(x)=\int_0^1 \id c \,[f(x_1,...,c \,(x_i+x_{i+1}),(1-c)\,(x_i+x_{i+1}),...,x_N)-f(x)]
\]
which, informally, means that $k(x,y)=1$ whenever $ y=(x_1,...,c \,(x_i+x_{i+1}),(1-c)\,(x_i+x_{i+1}),...,x_N).$
For the sites connected to the reservoirs,
\[L_{r} f\,(x)=\beta\,\int_0^\infty \id x'_N \,e^{-\beta\, x'_N}[f(y)-f(x)],\quad y=(x_1,\ldots,x_{N-1},x'_N),\]
\[L_\ell f\,( x)=\beta(1 + \ve)\,\int_0^\infty \id x'_1 \,e^{-\beta x'_1(1+\ve)}[f(y)-f(x)],\quad y=(x'_1,x_2,\ldots,x_N)\] 
where the temperature of the left bath gets perturbed with amplitude $\ve$ to enter the close-to-equilibrium regime: $\beta_\ell = \beta(1+\ve)$ and $\beta$ is the inverse temperature at the right edge.\\
The full stationary distribution is recently discussed in \cite{carinci2023solvable} and has a combinatorial interpretation; here we give the McLennan distribution which keeps the physical appearance as discussed above.\\

In \eqref{rate}, $\psi_\beta(x,y)= \beta e^{-\beta (x_\ell + y_\ell)/2}, w(x,y)=-\beta(y_\ell - x_\ell)$. From  \eqref{mca} and  \eqref{wp}, the McLennan distribution  is 
\begin{eqnarray}
\rho_{\text{ML}}(x) &=& \rho_{\text{eq}}(x)[1-  \,\ve\,\beta^2 \, \int_0^{\infty} \id t\,e^{t\,L_{\text {eq}}}\, \int_0^\infty \id E e^{-\beta E} \,(E-x_\ell)] \nonumber\\
&=& \rho_{\text{eq}}(x)[1 +  \,\ve\,\beta \, \int_0^{\infty} \id t\,e^{t\,L_{\text {eq}}}\, (x_\ell - \frac 1{\beta})]\nonumber\\
&=& \rho_{\text{eq}}(x)[1 +  \,\ve\,\beta \, \int_0^\infty\id t\,\big\langle x_\ell(t) - \frac 1{\beta}\,|\, x(0) =x\big\rangle_{\text {eq}}]
\end{eqnarray}
where the generator $L_{\text {eq}}$ for the equilibrium process where $\ve=0,$ acts on $x_\ell$.  The integrand $\big\langle x_\ell(t) - \frac 1{\beta}\,|\, x(0) =x\big\rangle_{\text {eq}}$ is the transient energy flux at time $t$ to the left reservoir, giving the previously stated thermal interpretation of the quasipotential to linear order around equilibrium. 
\end{example}

\section{Heat capacity}\label{heatcapacity}
The (nonequilibrium) heat capacity is operationally defined from the notion of excess heat; we refer to \cite{epl,activePritha,calo,Pe_2012} for theory and examples.  It can be obtained from the quasipotential \eqref{qs},
\begin{equation}\label{cb}
C(\beta) = \beta^2\big\langle \frac{\partial V}{\partial \beta}\big\rangle^s
\end{equation}
The expectation is in the nonequilibrium steady state for which we measure the thermal response.  Under detailed balance, for the equilibrium canonical ensemble, we have $V(x) = E(x) - \langle E\rangle_\text{eq}$, and hence the equilibrium heat capacity (at fixed volume) is $C_\text{eq}(\beta) = -\beta^2\frac{\partial \langle E\rangle_\text{eq}}{\partial \beta}$, as it should.  We are interested here in the close-to-equilibrium regime.  There, we have
\begin{equation}\label{CML}
C_\text{ML}(\beta) =  \beta\,\left(\langle E\rangle_\text{eq} - \langle E\rangle_\text{ML} \right) -\beta^2\frac{\partial}{\partial \beta}\langle E\rangle_\text{ML}  
\end{equation}
where the expectation $\langle\cdot\rangle_\text{ML}$ is over the McLennan distribution.  That is a big simplification with respect to \eqref{cb} as we do not need to evaluate the quasipotential directly: we just need to measure changes in energy; the heat capacity gets expressed in energy expectations.\\
Formula \eqref{CML} first appears as Eq.15 in \cite{epl}.  For consistency, we give the proof:\\
For formula \eqref{cb} we use $V_\text{ML} = E  - \langle E\rangle_\text{ML}  + W $ or $V = E  - \langle E\rangle_\text{ML}  + W +O(\ve^2)$, so that 
\begin{eqnarray}\label{cb1}
C(\beta) &=& \beta^2\langle \frac{\partial ( - \langle E\rangle_\text{ML} +  W )}{\partial \beta}\rangle_\text{ML} + O(\ve^2)\nonumber\\
&=& -\beta^2\,\frac{\partial\langle  E\rangle_\text{ML}}{\partial \beta} +\beta^2\langle \frac{\partial W }{\partial \beta}\rangle_\text{eq} + O(\ve^2)
\end{eqnarray}
On the other hand, since $\langle W \rangle_\text{eq} = O(\ve^2)$, we have $\langle \frac{\partial W }{\partial \beta}\rangle_\text{eq} = -\langle W\,\frac{\partial}{\partial \beta}\log \rho_\text{eq}\rangle_\text{eq} = \langle W\,E\rangle_\text{eq}$. 
Finally, use that $\beta \langle WE\rangle_\text{eq} = \langle E\rangle_\text{eq} - \langle E\rangle_\text{ML}$.\\

Furthermore, that close-to-equilibrium heat capacity can be derived from an entropy.  We indeed know that Clausius heat theorem remains valid close to equilibrium; there exists an exact differential for the excess heat over temperature; see \cite{Bertini_2012,Bertini_2013,kom,kom2,clau}.  The McLennan heat capacity in equation \eqref{CML},  can  indeed be   obtained from the McLennan entropy $S_{\text{ML}}=-\sum_x\rho_{\text{ML}}(x)\, \log\rho_{\text {ML}}(x)$: since $\langle V_{\text{ML}}\rangle_\text{ML}=0$ to leading order, and with $\rho_{\text {ML}}= \exp{-\beta V_\text{ML}}/{\cal Z}$,
\begin{align*}
    -\beta\, \frac{\partial \, S_{\text {ML}}}{\partial  \beta}
    &=\beta\,\frac{\partial \, }{\partial \beta}(-\beta\left\langle V_{\text{ML}}\right \rangle_{\text {ML}}-\log \cal Z)\\
    &=-\beta \,\frac{\partial \, }{\partial \beta}\log \cal Z\\
 &=-\beta (-\beta\,\langle\frac{\partial V_{\text{ML} }}{\partial \beta} \rangle_{\text {ML}}-\left\langle V_{\text{ML}}\right \rangle_{\text {ML}} )\\
    &=\beta^2\,\langle\frac{\partial V_{{\text{ML} }}}{\partial \beta} \rangle_{\text {ML}}
\end{align*}
which reproduces \eqref{cb}.

\section{Boundedness of the quasipotential}\label{bv}
The Third Law of thermodynamics, {\it aka} the Nernst postulate, implies that in equilibrium the heat capacity is going to zero as temperature goes to zero. It is worth remembering that the ``Law'' knows exceptions; equilibrium systems with highly degenerate ground states may not satisfy it, \cite{Aizenmanlieb}.   In recent work, \cite{jchemphys,mathnernst}, an extended Third Law was derived where the extension covers nonequilibrium jump processes: the nonequilibrium heat capacity vanishes at zero ambient temperature if the quasipotential remains bounded (in addition to having  a nondegenerate zero--temperature regime for its static fluctuations).  Here we show that in the McLennan-regime, $W$ and hence the quasipotential $V_\text{ML} = E  - \langle E\rangle_\text{ML}  + W$   always remains bounded as $\beta\uparrow \infty$.\\

To show the boundedness of the quasipotential $W$ in \eqref{mca}, we need some definitions from graph theory.\\
 Suppose a random walker is jumping on a connected graph denoted by $G (\cal V, \cal E)$. Where $\cal V$ represents the set of all vertices as the physical states, and $\cal E$ represents the set of all edges.  An edge $\{x,y\}\in \cal E$ exists if either $k(x,y)$ or $k(y,x)$ is not zero. Conversely,  if $k(x,y)=k(y,x)=0$, then there is no edge connecting state $x$ and $y$. \\
We denote the set of all spanning trees by $\cal T $. A rooted spanning tree is a spanning tree in which all edges are directed towards a specific vertex, referred to as the root. The set of all spanning trees rooted in vertex $x$ is denoted by $\cal T_x$. $\tau(x\to y)$  represents a unique path that starts from vertex $x$ and ends at vertex $y$ within the tree $\tau$. The weight of a rooted spanning tree such as $\tau_x$ is represented by $m (\tau_x)$, which is calculated by taking the product of the rates assigned to all edges on the tree,
\[ m(\tau_x):=\prod_{(u,u')\in \tau_x}k(u,u') \]
 We also define 
\[ m(x):=\sum_{\tau_x\in \cal T_x}m(\tau_x),\qquad m:=\sum_x m(x) \]
and if the rates satisfy detailed balance, we write $m_{\text{eq}}$, $m_{\text{eq}}(x)$ and $m_{\text{eq}}(\tau_x)$.\\

Here we use an  expression  for $\frac{1}{\beta}(\frac{\rho_\text{ML}}{\rho_\text{eq}}-1)$ which was obtained\footnote{In  \cite{intro}, different notations are used: here we use $ \ve [\psi_\beta(x,y) + \frac{\beta}{2}w(x,y)]$ while in \cite{intro} it was $\frac{\ve}{2}\, s(x,y)$.} as relation  (5.3) in \cite{intro} (and its appendix): 
\begin{align}\label{gw}
     W(x)&=\ve \,  \sum_z \rho_\text{eq}(z)\frac{\sum_{\tau_z\in \cal T_z} \,m_{\text{eq}}(\tau_z) \omega_\tau(x\to z)}{m_{\text{eq}}(z)}\notag\\
     &=\ve \,\frac{1}{m_{\text{eq}}}\,  \sum_z\sum_{\tau_z\in \cal T_z} \,m_{\text{eq}}(\tau_z) \omega_\tau(x\to z)
\end{align}
where according to  the Kirchhoff formula (see also \cite{intro}), 
\[
\rho_\text{eq}(z)=\dfrac{m_{\text{eq}(z)}}{m_{\text{eq}}}\qquad \omega_\tau(x\to z):=\sum _{(u,u')\in \tau(x\to z)}w(u,u')
\]
For being self-consistent,  we prove in Appendix \ref{lw} that the graphical expression \eqref{gw} is indeed satisfying the Poisson equation \eqref{cre}.

\begin{proposition}
    The McLennan quasipotential $V_\text{ML}$ is uniformly bounded in temperature.
    \end{proposition}
    \begin{proof}
    From \eqref{gw}, and as $V_\text{ML} = E  - \langle E\rangle_\text{ML}  + W$, the graphical representation of the quasipotential is 
    \begin{equation}\label{gv}
        V_\text{ML}(x)= E(x)  - \langle E\rangle_\text{ML}  +\ve \,\frac{1}{m_{\text{eq}}}\, \sum_{z\in \cal V} \sum_{\tau_z\in \cal T_z} \,m_{\text{eq}}(\tau_y) \omega_\tau(x\to y).
    \end{equation}
The ratio $\dfrac{m_{\text{eq}} (\tau_y)}{m_{\text{eq}}}$ is bounded for all $y$, 
    \[ 0\leq\frac{m_{\text{eq}} (\tau_y)}{m_{\text{eq}}} \leq 1\]
Hence, for every $y$ and $\tau $,
\[\lim_{\beta \to \infty}\frac{m_{\text{eq}} (\tau_u)\, \omega_\tau(x\to y)}{m_{\text{eq}}}\]
is bounded.
 \end{proof}

\section{Examples}\label{example}
The present section optimistically explores the possible validity of the close-to-equilibrium heat capacity \eqref{CML} to regimes far from equilibrium.  That includes the low-temperature regime  (as we multiply the driving $\ve$ with the inverse temperature $
\beta$ in the transition rates.)
We will see how the close-to-equilibrium heat capacity very well approximates the exact nonequilibrium heat capacity at sufficiently high temperatures but also at intermediate values.\\

In general, nonequilibrium heat capacities may be negative.  As explained in \cite{jchemphys,prithaarray}, it expresses an anticorrelation between heat and the occupation statistics.  Yet, close-to-equilibrium that is impossible.  The reason is exactly the fact that, as in equilibrium, heat and occupation are given by the same (quasi)potential. For McLennan, that is the combination of \eqref{mc} with \eqref{cb}. Since the heat capacity is always positive in the canonical ensemble (and is proportional to the variance of the energy),  a small perturbation cannot change that.\\

\begin{example}[merry-go-round]
Consider the graph in Fig.~\ref{abc}. 

\begin{figure}[H]
    \centering
    \includegraphics[scale=0.28]{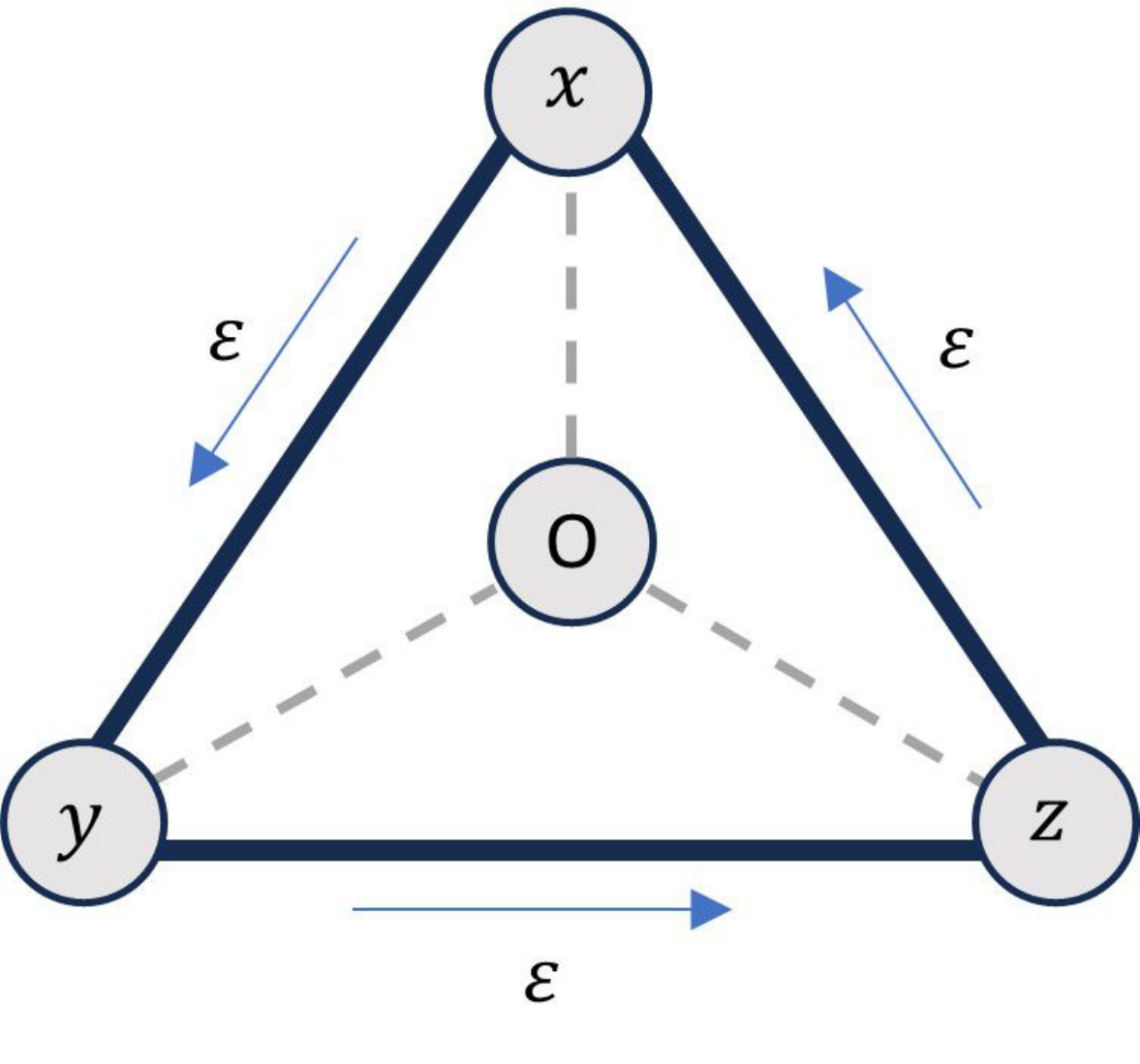}
    \caption{\small{A merry-go-round graph with three states $x\leftrightarrow y\leftrightarrow z \leftrightarrow x$ on a loop, and possible transitions to a sink/source state $o$ in the center.}}
   \end{figure} \label{abc}
   
The transition rates between the corner states and the central state are  $k(x,o)=k(y,o) = k(z,o)=e^{\frac{\beta}{2}\mu}e^{-\beta \Delta}$,  $
k(o,x)=k(o,y) = k(o,z)=e^{-\frac{\beta}{2}\mu}e^{-\beta \Delta} $ and the rates on the merry-go-round are
$$ k(x,y)=k(y,z)=k(z,x)=\frac{1}{1+e^{-\beta \ve}}, \quad k(x,z) =k(z,y)=k(y,x)= \frac{1}{1+e^{\beta  \ve}}.$$  The driving is $\ve$ giving a counter-clockwise bias.\\
From $\eqref{ldb}$, the  expressions for the exchanged heat are
\[q(x,o)=q(y,o)=q(z,o)=\mu,\qquad q(x,y)=q(y,z)=q(z,x)=\ve\]
For fixed $\beta$, to second order in $\ve$, the quasipotential is
\begin{eqnarray}
    &V(x)=V(y)=V(z)=\left(\mu -\frac{3 \,\mu }{e^{\beta  \,\mu \,}+3}\right)+\frac{\beta \, e^{\frac{1}{2} \beta\,  (\Delta +3\, \mu )}}{2 \left(e^{\beta \, \mu }+3\right)^2}\, \varepsilon ^2 +O\left(\varepsilon\, ^3\right)\\
 & V(o)=-\frac{3 \,\mu }{e^{\beta\,  \mu }+3}-\frac{3\, \, \beta  \,e^{\frac{1}{2} \beta  (\Delta +\mu )}}{2\left(e^{\beta  \mu }+3\right)^2}\, \varepsilon ^2+O\left(\varepsilon ^3\right)  
\end{eqnarray}
Note that, by symmetry, there is no linear order term in $\ve$.  It is the second-order term that can diverge for $\beta\uparrow$. Indeed, the quasipotential is not bounded for $\beta\uparrow$ when $\Delta$ is large with respect to $|\mu|$.\\
See \fig \ref{merryHC1} for a comparison between the equilibrium heat capacity $C_\text{eq}$ and the true nonequilibrium heat capacity $C$.  The latter depends on  $\ve$.  Clearly,  the nonequilibrium becomes more pronounced at low temperatures since $\ve$ is multiplied with $\beta$ in the transition rates.

\begin{figure}[H]
	\begin{subfigure}{0.49\textwidth}
         \centering
         \def\svgwidth{0.8\linewidth}
		\includegraphics[scale=0.85]{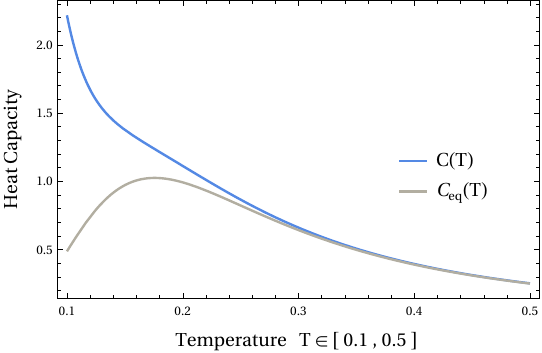}  
	\end{subfigure}
  \hfill
\begin{subfigure}{0.49\textwidth}
         \centering
         \def\svgwidth{0.8\linewidth}
		\includegraphics[scale=0.85]{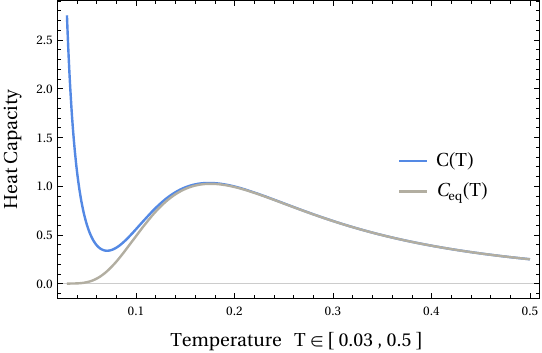}  
	\end{subfigure}
\caption{\small{Heat capacity of merry-go-round as summarized in Fig.~\ref{abc}. $\Delta=1.5,\, \mu=0.5 $,  left:  $\ve=0.05 $ and right:  $\ve=0.01 $.  For smaller $\ve$ the two curves stay together longer as $T\downarrow 0$.}}\label{merryHC1}
\end{figure}

To get a nontrivial linear regime for the heat capacity, we take another version of the merry-go-round by introducing energies for the corner states $x,y,z$. We now take
the rates between the corner states and the central state as
\begin{align*}
k(x,o)&=e^{\frac{\beta}{2}(\mu+E)}e^{-\beta \Delta}, \quad k(y,o) = e^{\frac{\beta}{2}(\mu+2E)}e^{-\beta \Delta}, \quad k(z,o)=e^{\frac{\beta}{2}(\mu+3E)}e^{-\beta \Delta}\\
k(o,x)&=e^{-\frac{\beta}{2}(\mu+E)}e^{-\beta \Delta}, \quad k(o,y) = e^{-\frac{\beta}{2}(\mu+2E)}e^{-\beta \Delta}, \quad k(o,z)=e^{-\frac{\beta}{2}(\mu+3E)}e^{-\beta \Delta}
\end{align*}
The rates on the merry-go-round change to
\begin{align*}
   k(x,y)&=k(y,z)=\frac{1}{1+e^{-\beta (\ve-E)}}, \quad k(z,x)=\frac{1}{1+e^{-\beta (2E+\ve)}},\\
 k(z,y)&=k(y,x)= \frac{1}{1+e^{\beta ( \ve-E)}} , \quad    k(x,z)= \frac{1}{1+e^{\beta (2E+\ve)}}
\end{align*}
From $\eqref{ldb}$, the corresponding heat over each transition becomes
\begin{align*}
q(x,o)=\mu+E,\quad q(y,o)=\mu+2E,\quad q(z,o)= \mu+3E\\
q(x,y)=\ve-E,\quad q(y,z)=\ve-E,\quad q(z,x)=\ve +2E
\end{align*}
and there is now a linear order in $\ve$ for the quasipotential.\\
Fig.~\ref{merryHC2} summarizes the situation for the heat capacities, comparing the full nonequilibrium $C(T)$ with the McLennan approximation $C_\text{ML}(T)$.
\begin{figure}[H]
	\begin{subfigure}{0.49\textwidth}
         \centering
         \def\svgwidth{0.8\linewidth}
		\includegraphics[scale=0.85]{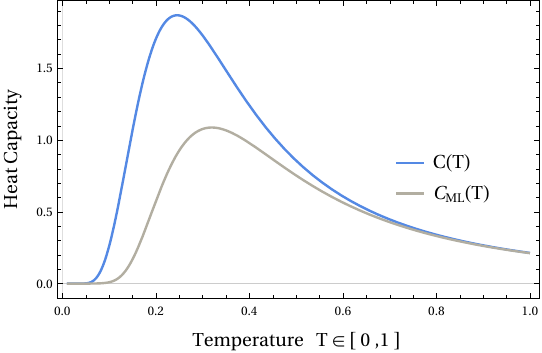}  
	\end{subfigure}
  \hfill
\begin{subfigure}{0.49\textwidth}
         \centering
         \def\svgwidth{0.8\linewidth}
		\includegraphics[scale=0.85]{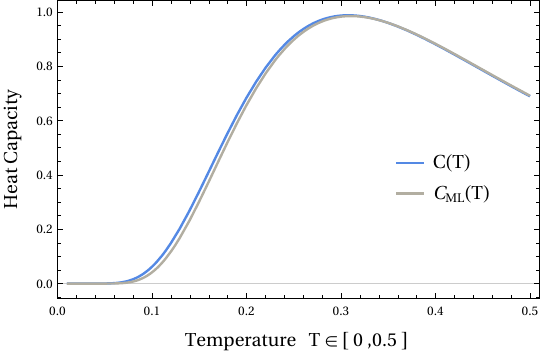}  
	\end{subfigure}
\caption{\small{Heat capacity of Fig.~\ref{abc}. $\Delta=1.5,\, \mu=0.5,\, E=0.2 $,  left:  $\ve=0.5 $ and right:  $\ve=0.05 $.}}\label{merryHC2}
\end{figure}
This time, there is no divergence of the heat capacity at zero temperature, and the approximation works very well for small $\ve$.
\end{example} 

\begin{example}[Loop-tree]\label{ltr}
Consider a random walker jumping on the graph in Fig.~\ref{badhair}. The transition rates for the jumps $u\rightarrow u'$ are 
\[
k(u,u')=\frac{1}{1+e^{-\beta(E(u)-E(u')+s(u,u'))}},\qquad s(u,u')=-s(u',u)
\]
where $s(x,z)=s(z,y)=s(y,x)=\ve$, $s(z,r)=s(r,f)=0$, and the energy $E$ of each state is given in Fig.~\ref{badhair}.

\begin{figure}[ht!]
    \centering
    \includegraphics[scale=0.3]{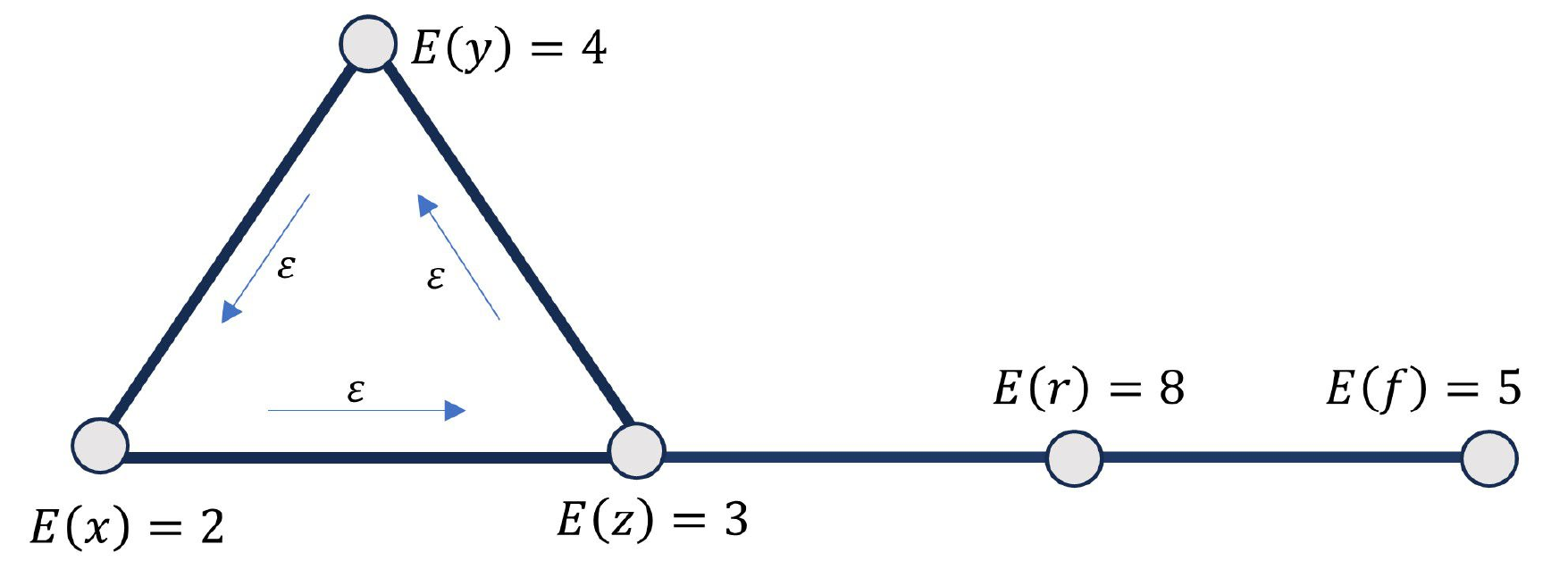}
    \caption{A loop connected with a line representing Example~\ref{ltr}.}
    \label{badhair}
\end{figure}

The heat exchanged during the transitions is obtained from \eqref{ldb},
\[q(x,z)=-1+\ve, \,\, q(z,y)=-1+\ve, \,\, q(y,x)=2+\ve,\,\, q(z,r)=-5,\,\, q(r,f)=3\]
The McLennan quasipotential is \begin{align*}
    V_{\text{ML}}(x)&=-\frac{1}{a} (3e^{3 \beta }+2 e^{4 \beta }+e^{5 \beta }+6)-d\,\ve\\
    V_{\text{ML}}(y)&= V_{\text{ML}}(x)+[2-\frac{\left(e^{\beta }-1\right)^2 \varepsilon }{2 \left(e^{\beta }+e^{2 \beta }+1\right)}]\\
    V_{\text{ML}}(z)&= V_{\text{ML}}(x)+[1-\frac{\left(-e^{\beta }+2 e^{2 \beta }-1\right) \varepsilon }{2 \left(e^{\beta }+e^{2 \beta }+1\right)}]\\
    V_{\text{ML}}(r)&= V_{\text{ML}}(x)+[6-\frac{\left(-e^{\beta }+2 e^{2 \beta }-1\right) \varepsilon }{2 \left(e^{\beta }+e^{2 \beta }+1\right)}]\\
    V_{\text{ML}}(f)&= V_{\text{ML}}(x)+[3-\frac{\left(-e^{\beta }+2 e^{2 \beta }-1\right) \varepsilon }{2 \left(e^{\beta }+e^{2 \beta }+1\right)}]
    \end{align*}
for $a=e^{3 \beta }+e^{4 \beta }+e^{5 \beta }+e^{6 \beta }+1$
and
\begin{align*}
   d&=\frac{e^{\beta }-1}{2 a^2 \,(e^{\beta }+e^{2 \beta }+1)} \bigg(8 e^{4 \beta } \beta +4 e^{5 \beta } \beta +6 e^{6 \beta } \beta +14 e^{7 \beta } \beta +e^{8 \beta } \beta +e^{9 \beta } \beta +3 e^{10 \beta } \beta\\
   &+3 e^{11 \beta } \beta +2 e^{12 \beta } \beta -2 e^{\beta }-2 e^{3 \beta }-4 e^{4 \beta }-5 e^{5 \beta }-6 e^{6 \beta }-4 e^{7 \beta }-4 e^{8 \beta }\\
   &-6 e^{9 \beta }-5 e^{10 \beta }-4 e^{11 \beta }-2 e^{12 \beta }-1\bigg)
\end{align*}

In Fig.\ref{hcmlbadhair} left, we plot the quasipotentials $V$ defined in \eqref{pe}, together with their 
McLennan approximations $V_{ML} = E - \langle E\rangle_\text{ML} + W$ as computed above, for states $x,y,z$. The true quasipotential $V$ is not bounded towards zero temperature. However, the  McLennan approximations $V_{ML}$ are uniformly bounded for all states. 


The McLennan heat capacity  \eqref{CML} and the (full) heat capacity  \eqref{cb} are plotted in \fig\ref{hcmlbadhair} right.   For small $\ve$, $C_{\text{ML}}(T)$ tends to overlap with $C(T)$, even at low temperatures.  
\begin{figure}[H]
	\begin{subfigure}{0.49\textwidth}
         \centering
         \def\svgwidth{0.8\linewidth} 
		\includegraphics[scale=0.8]{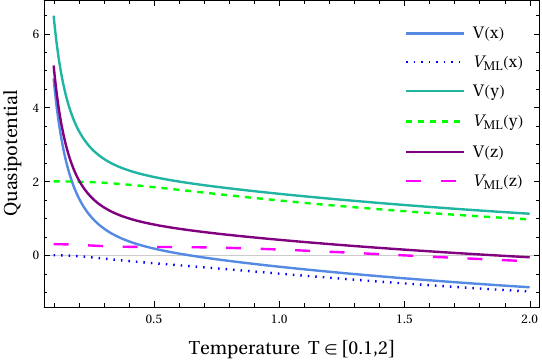}  
	\end{subfigure}
  \hfill
	\begin{subfigure}{0.49\textwidth}
         \centering
         \def\svgwidth{0.8\linewidth} 
		\includegraphics[scale=0.8]{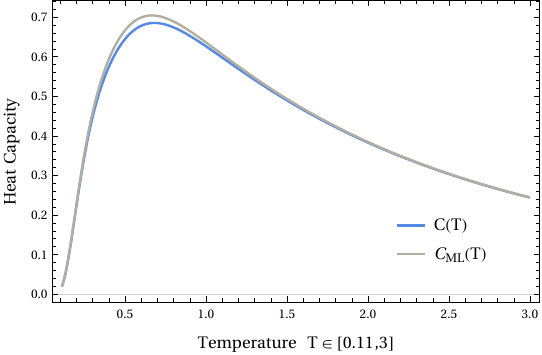}  
	\end{subfigure}
    \caption{\small{Left: quasipotentials $V$ and $V_\text{ML}$ of states $x, y, z$, for the dynamics represented in Fig.~\ref{badhair} with  $\ve=0.7$. Right:  heat capacity for same model with $\ve=0.1$. }}
	\label{hcmlbadhair}
\end{figure}
\end{example}

\section{Conclusions}\label{con}
The thermal response around equilibrium is coded by the McLennan ensemble via the so-called quasipotential.  The quasipotential is a measure of excess heat on the one hand and, on the other hand, yields the occupation statistics.  It therefore carries an extension of the useful relation between fluctuations and dissipation as we know it in equilibrium.\\
The resulting McLennan heat capacity (which is the heat capacity close to equilibrium) can be given entirely in terms of energy expectations (extending the even simpler formula for canonical equilibrium), and is derivable from an entropy.  As the linear part around equilibrium of the nonequilibrium heat capacity, that McLennan heat capacity is always nonnegative and tends to zero with vanishing temperature.  The breaking of the extended Third Law is therefore a nonlinear effect, starting only at second order. \\
We have added simple examples to illustrate these key-points.

\appendix
\section{Graphical solution of Poisson equation}\label{lw}

 \begin{lemma}\label{w}
 $W$ in \ref{gw} is satisfying the equation:
 \[L_{\text{eq}}W(x)= -\ve \sum_yk_{\text{eq}}(x,y) w(x,y)\]
\end{lemma}
 \begin{proof}
\begin{align*}
   L_{\text{eq}}W(x)&=\ve \sum_y k_{\text{eq}}(x,y)(W(y)-W(x))\\
   &=\frac{\ve }{m_{\text{eq}}}\sum_y k_{\text{eq}}(x,y)\, \sum_{z\in \cal V}\,\sum_{\tau_z\in \cal T_z} \, \,m_{\text{eq}}(\tau_z) \big( \omega_\tau(y\to z) -  \omega_\tau(x\to z)\big)\\
   &=\frac{\ve }{m}\sum_y k_{\text{eq}}(x,y)\, \sum_{z\in \cal V}\,\sum_{\tau_z\in \cal T_z} \, \,m_{\text{eq}}(\tau_z) \omega_\tau(y\to x)\\
   &=\frac{\ve }{m_{\text{eq}}}\sum_y k_{\text{eq}}(x,y)\, \sum_{z\in \cal V} \big(\,\sum_{\tau_z \ni(y,x)} \, \,m_{\text{eq}}(\tau_z)\omega_\tau(y\to x)+\sum_{\tau_z \not\ni(y,x)} \, \,m_{\text{eq}}(\tau_z)\omega_\tau(y\to x)\big) 
\end{align*}
where in the third line the antisymmetry of $w(u,u')$ is used. If the tree $\tau$ includes the edge $(y,x)$, then the path between $y$ to $x$ within the tree $\tau$ is exactly the edge $(y,x)$, implying $\omega_\tau(y\to x)=w(y,x)$ and hence,
\begin{align}\label{proof1}
   L_{\text{eq}}W(x)&=\frac{\ve }{m_{\text{eq}}}\sum_y k_{\text{eq}}(x,y)\, \sum_{z\in \cal V} \big(\,\sum_{\tau_z \ni(y,x)} \, \,m_{\text{eq}}(\tau_z)w(y,x)+\sum_{\tau_z \not\ni(y,x)} \, \,m_{\text{eq}}(\tau_z)\omega_\tau(y\to x)\big) \notag\\
   &=\ve \sum_y k_{\text{eq}}(x,y)w(y,x)\, \frac{1}{m_{\text{eq}}}\sum_{z\in \cal V} \big(\,\sum_{\tau_z \ni(y,x)} \, \,m_{\text{eq}}(\tau_z)+\sum_{\tau_z \not \ni(y,x)} \, \,m_{\text{eq}}(\tau_z)\big)\\
   &\quad+\frac{\ve }{m_{\text{eq}}}\sum_y k_{\text{eq}}(x,y) \sum_{z\in \cal V}\big(-\sum_{\tau_z \not \ni(y,x)} \, \,m_{\text{eq}}(\tau_z)w(y,x)+\sum_{\tau_z \not\ni(y,x)} \, \,m_{\text{eq}}(\tau_z)\omega_\tau(y\to x) \big)\notag\\
   &=-\ve \sum_y k_{\text{eq}}(x,y)w(x,y)+\frac{\ve }{m_{\text{eq}}}\sum_{y,z} k_{\text{eq}}(x,y) \sum_{\tau_z \not \ni(y,x)} \, \,m_{\text{eq}}(\tau_z)\big(w(x,y)+\omega_\tau(y\to x) \big)\notag
\end{align}
We show that the second sum in the last line equals zero.\\

\textit{If $z=x$}, fix a spanning tree  $\tau $ such that  $ (x,y)\not \in \tau_x$. There is a path between $x$ and $y$ in the tree, and adding the edge $(x,y)$ to $\tau_x$ creates a loop with some rooted tree connected to the loop. That loop is made by the path  $\tau(y\to x)$ together with the edge $(x,y)$. The vertex $x$ has another neighbor in that loop that we call $y'$. A spanning tree exists made by removing the edge $(y',x)$ from the tree $\tau$ and adding the edge $(y,x)$; we call it $\tau'$. Now, adding the edge $(x,y')$ to the spanning tree $\tau'_z$ makes a loop with some rooted tree connected to the loop. That loop is made by the path $\tau'(y'\to x)$ together with $(x,y')$. As a result, the graphs $\tau_x \cup (x,y)$ and $\tau'_x\cup (x,y')$ have different directions in the loop, but the rest is similar; see \fig \ref{taux}.
\begin{figure}[H]
     \centering
     \begin{subfigure}{0.49\textwidth}
         \centering
         \def\svgwidth{0.8\linewidth}        
         \includegraphics[scale = 0.27]{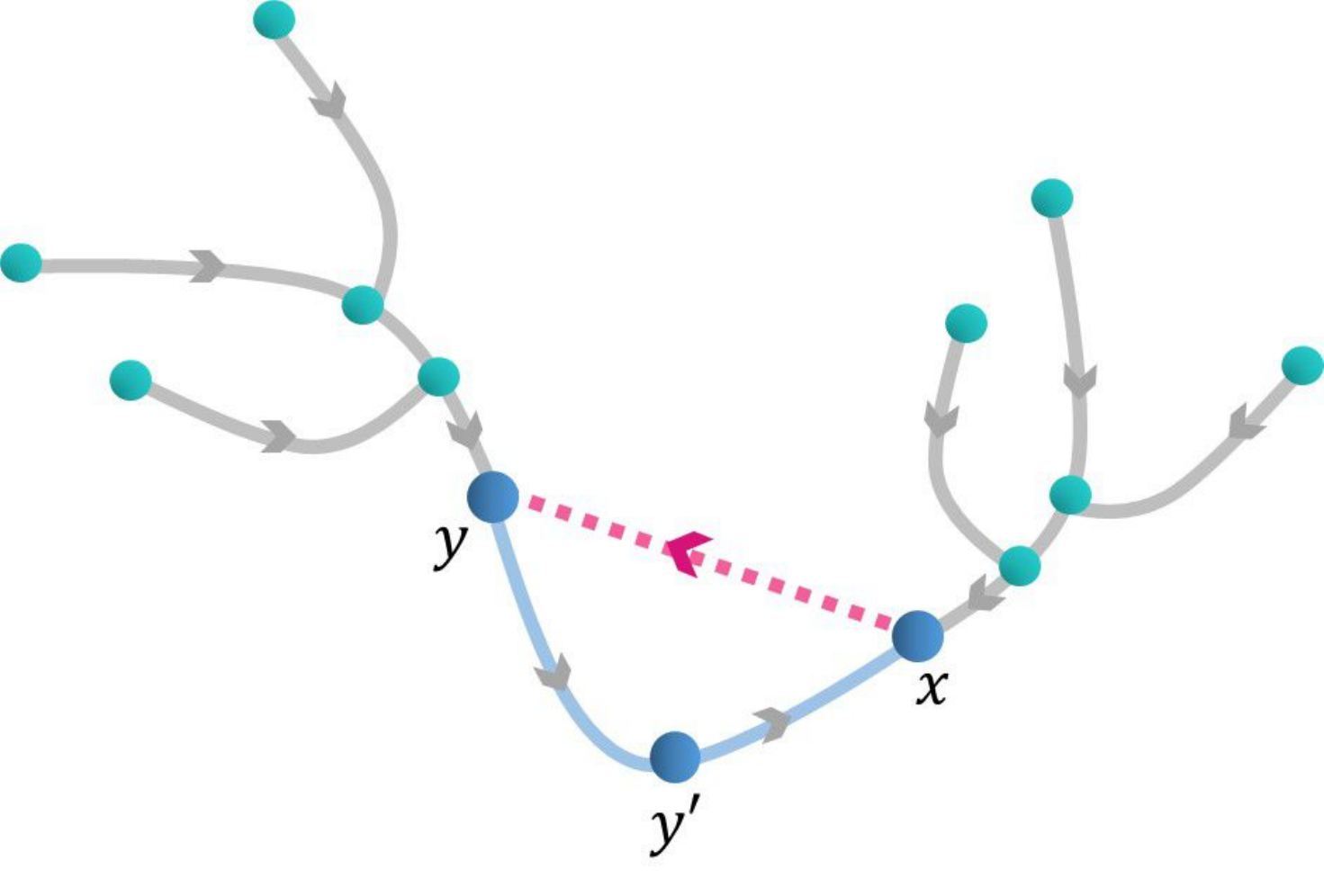}
         \caption{$\tau_x \cup (x,y)$}
     \end{subfigure}
     \hfill
     \begin{subfigure}{0.49\textwidth}
         \centering
         \def\svgwidth{0.8\linewidth}        
\includegraphics[scale = 0.27]{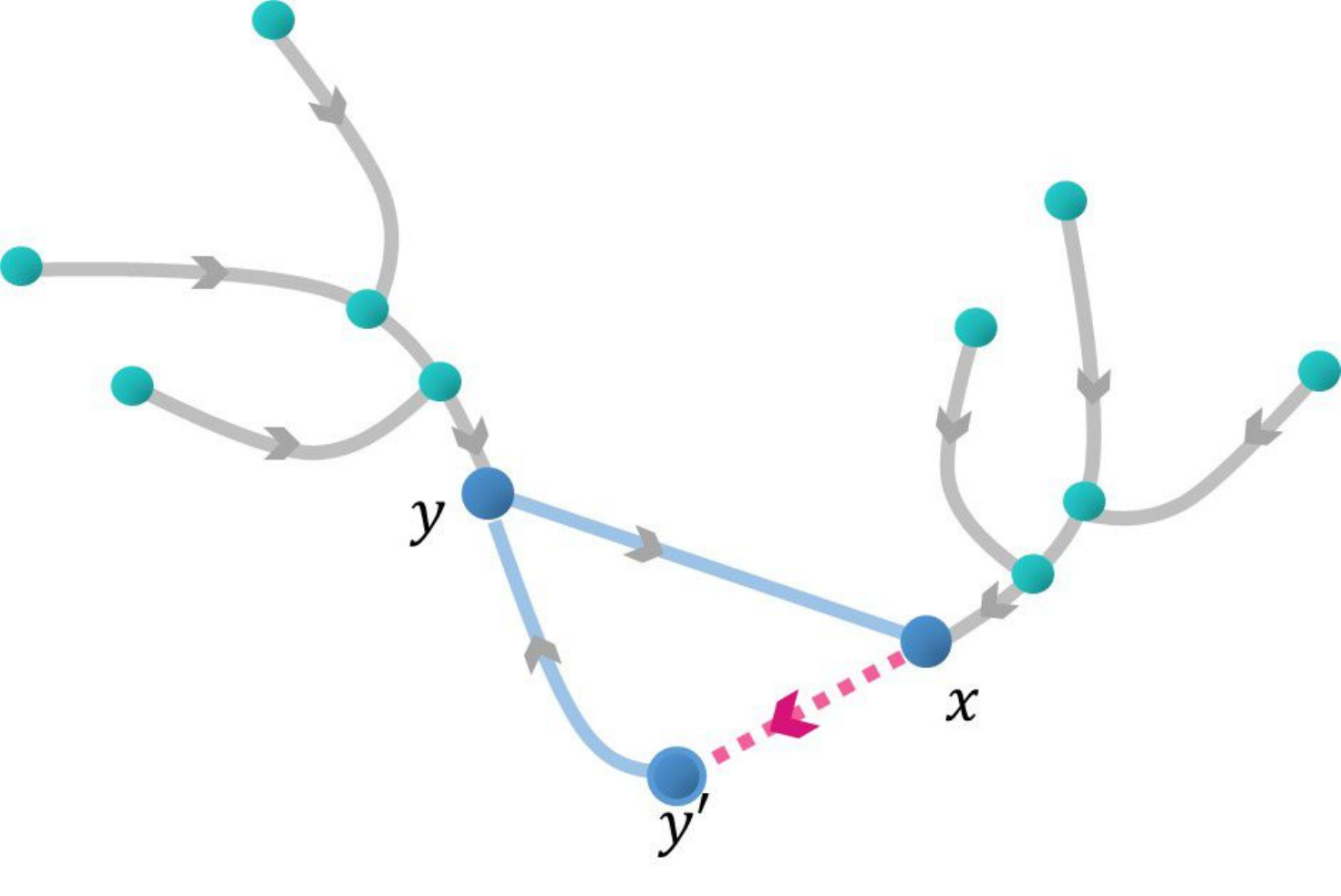}
         \caption{$\tau'_x\cup (x,y')$}
     \end{subfigure}
        \caption{\small{  Adding  $(x,y)$ to  $\tau_x  $ and the edge $(x,y')$ to  $\tau'_x$ makes two similar graphs: (a) and (b) are only different in the direction of their unique loop. }}\label{taux}
\end{figure}
From detailed balance, two different directions of a loop have the same weight and
\[m_{\text{eq}}(\tau_x)\, k_{\text{eq}}(x,y)=m_{\text{eq}}(\tau'_x)\, k_{\text{eq}}(x,y')\]
Therefore, by the antisymmetry of $w(x,y)$,
\begin{align*}
 &m_{\text{eq}}(\tau_x)\, k_{\text{eq}}(x,y)(w(y,x)+\omega_\tau(x\to y)) +   m_{\text{eq}}(\tau'_x)\, k_{\text{eq}}(x,y')(w(y',x)+\omega_{\tau'}(x\to y'))= 0
\end{align*}
where $w(y,x)+\omega_\tau(x\to y)$ and $w(y',x)+\omega_{\tau'}(x\to y')$ are sums over the $w(u,u')$ of the edges of the same loop but in two different directions. Hence,
\begin{equation}\label{proof2}
    \frac{1}{m_{\text{eq}}}\sum_{y} k_{\text{eq}}(x,y) \sum_{\tau_x\not \ni(y,x)} \, \,m_{\text{eq}}(\tau_x)\big(w(x,y)+\omega_\tau(y\to x)\big)=0
\end{equation}

\textit{If $z\not=x$}, take a spanning tree $\tau$ such that $(x,y)\not \in \tau$. There is a path between $x$ and $y$ in the tree.  Adding the edge $(x,y)$ to $\tau_z$ makes a loop with some rooted tree connected to the loop. That loop is made by the path  $\tau(y\to x)$ together with the edge $(x,y)$, but the edges on that loop are not oriented in one direction. The vertex $x$ has another neighbor in this loop that we call $y'$. A spanning tree exists made by removing the edge $(y',x)$ from the tree $\tau$ and adding the edge $(y,x)$; we call it $\tau'$. Now, adding the edge $(x,y')$ to the spanning tree $\tau'_z$ makes a loop with some rooted tree connected to the loop. This loop is made by the path $\tau'(y'\to x)$ together with $(x,y')$.  The graphs $\tau_x \cup (x,y)$ and $\tau'_x\cup (x,y')$ are the same; see \fig \ref{tauz}.
\begin{figure}[H]
     \centering
     \begin{subfigure}{0.49\textwidth}
         \centering
         \def\svgwidth{0.8\linewidth}        
         \includegraphics[scale = 0.27]{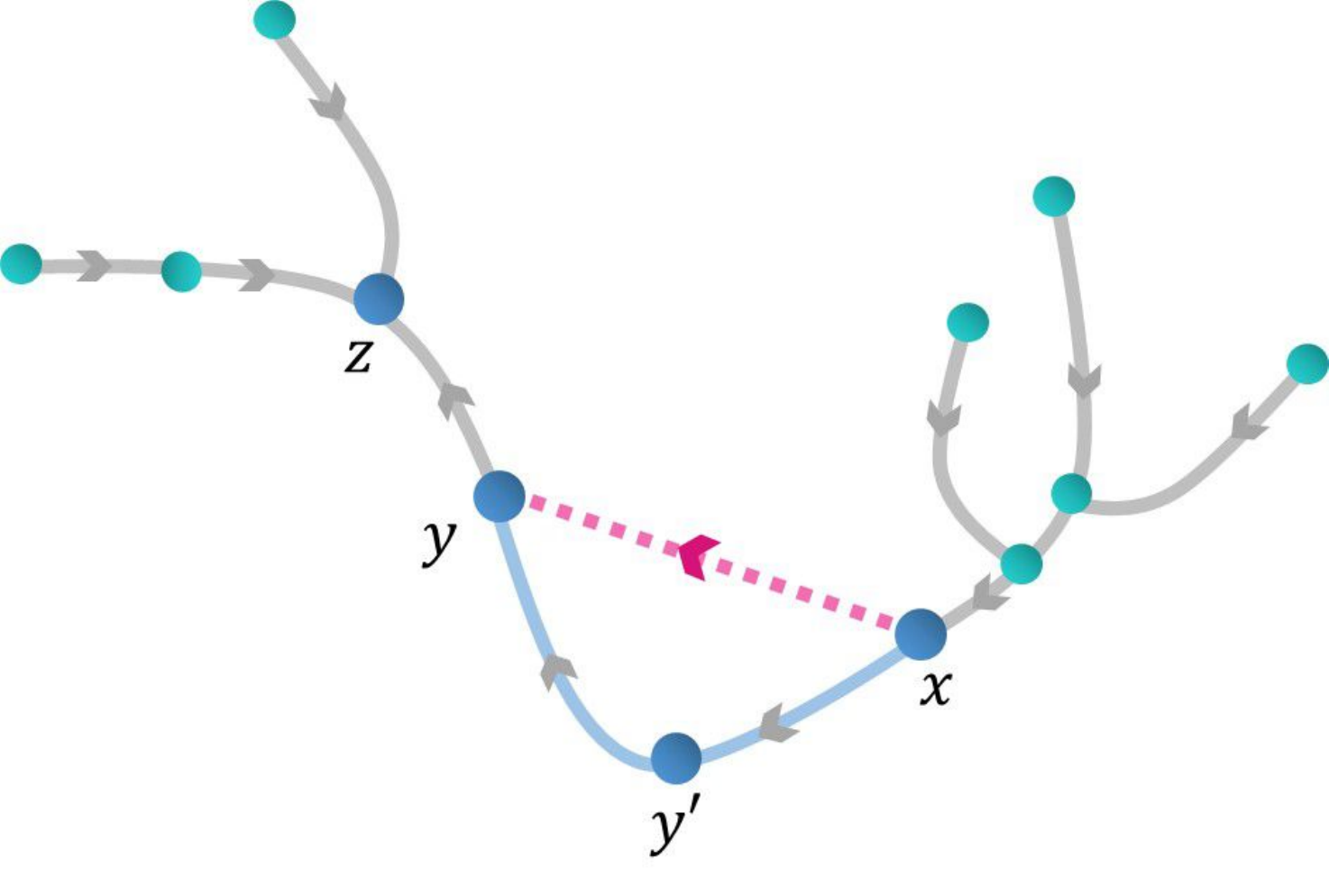}
         \caption{$\tau_z \cup (x,y)$}
     \end{subfigure}
     \hfill
     \begin{subfigure}{0.49\textwidth}
         \centering
         \def\svgwidth{0.8\linewidth}        
\includegraphics[scale = 0.27]{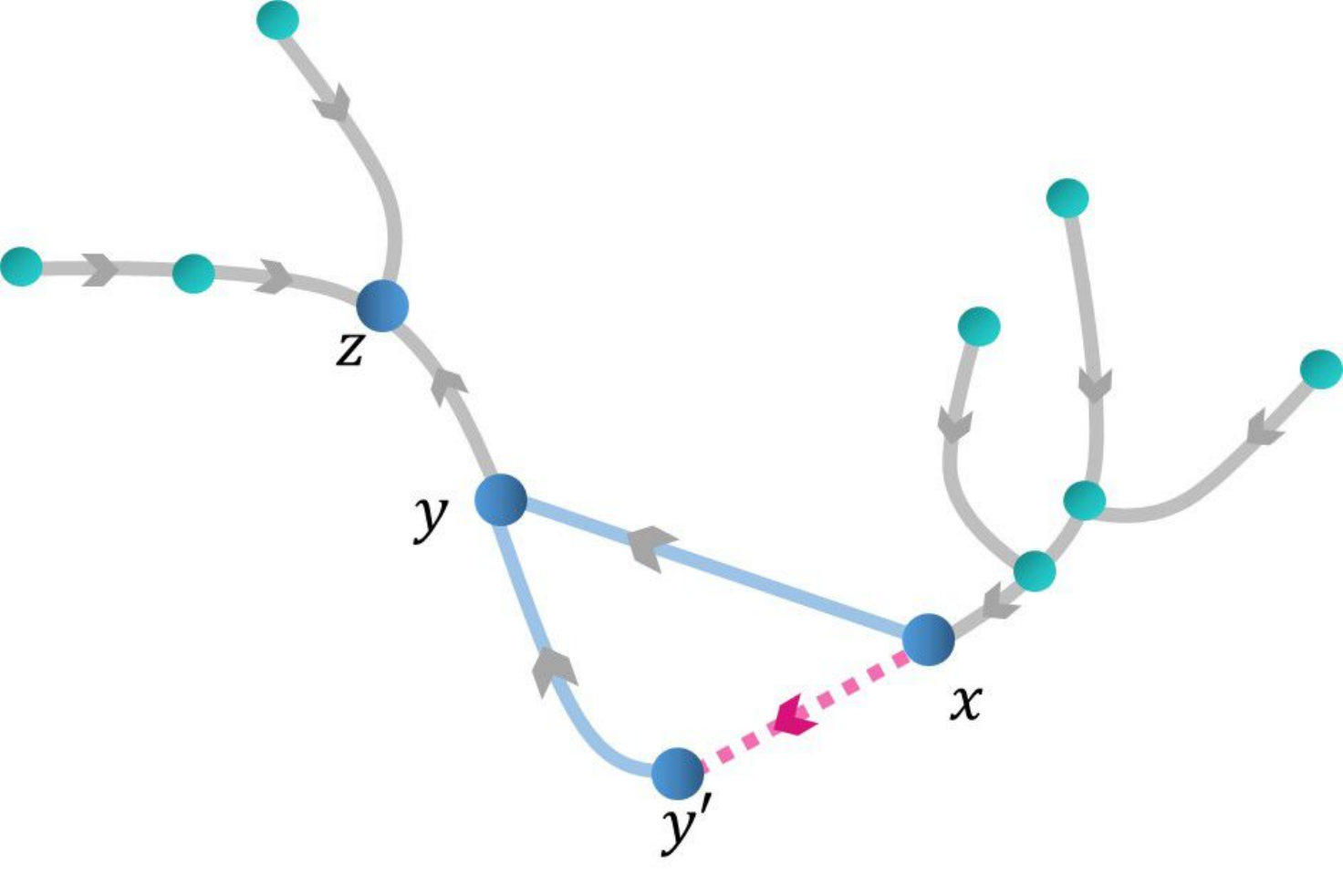}
         \caption{$\tau'_z\cup (x,y')$}
     \end{subfigure}
        \caption{\small{  Adding  $(x,y)$ to  $\tau_z  $ and  the edge $(x,y')$ to  $\tau'_z$ produces the same graphs (a)=(b). }}\label{tauz}
\end{figure}
Hence,
\begin{align*}
 &m_{\text{eq}}(\tau_z)\, k_{\text{eq}}(x,y)(w(y,x)+\omega_\tau(x\to y)) +   m_{\text{eq}}(\tau'_z)\, k_{\text{eq}}(x,y')(w(y',x)+\omega_{\tau'}(x\to y'))= 0
\end{align*}
where $w(y,x)+\omega_\tau(x\to y)$ and $w(y',x)+\omega_{\tau'}(x\to y')$ are sums over all $w(u,u')$ over the edges of the same loop in two different directions. 
We conclude that
\begin{equation}\label{proof3}
    \frac{1}{m_{\text{eq}}}\sum_{y,z} k_{\text{eq}}(x,y) \sum_{\tau_z\not \ni(y,x)} \, \,m_{\text{eq}}(\tau_z)\big(w(x,y)+\omega_\tau(y\to x)\big)=0
\end{equation}
and the proof is finished.
 \end{proof}

\bibliographystyle{unsrt}  
\bibliography{chr}
\onecolumngrid
\end{document}